\documentclass[11pt]{article}
\usepackage{amsmath,amssymb,amsthm}
\usepackage{fullpage}
\usepackage{enumitem}
\usepackage{setspace}
\usepackage{xcolor}
\usepackage[longnamesfirst]{natbib}

\newtheorem{definition}{Definition}
\newtheorem{assumption}{Assumption}
\newtheorem{theorem}{Theorem}

\newtheorem{corollary}{Corollary}

\allowdisplaybreaks
\begin{document}

\title{GMM and M Estimation under Network Dependence\thanks{July 18, 2025. An earlier version of this paper was circulated as arXiv:2503.00290v1, titled ``Uniform Limit Theory for Network Data'' (March 1, 2025). It has since been superseded by the present paper.}}
\author{Yuya Sasaki\thanks{Brian and Charlotte Grove Chair and Professor of Economics. Department of Economics, Vanderbilt University, PMB \#351819, 2301 Vanderbilt Place, Nashville, TN 37235-1819. Email: yuya.sasaki@vanderbilt.edu. I thank Brian and Charlotte Grove for their generous research support.}\\Vanderbilt University.}
\date{}
\maketitle

\begin{abstract}
This paper presents GMM and M estimators and their asymptotic properties for network-dependent data. To this end, I build on \citet*{KMS2021} and develop a novel uniform law of large numbers (ULLN), which is essential to ensure desired asymptotic behaviors of nonlinear estimators \citep[e.g.,][Section 2]{NeweyMcFadden1994}. Using this ULLN, I establish the consistency and asymptotic normality of both GMM and M estimators. For practical convenience, complete estimation and inference procedures are also provided.


\end{abstract}

\section{Introduction}\label{sec:intro}

In recent years, asymptotic analysis of network-dependent data has garnered significant attention in econometrics \citep*[e.g.,][]{kuersteiner2019limit,leung2019normal,kuersteiner2020dynamic,KMS2021}.\footnote{See also \citet{jenish2012spatial} for related work in which the dependence is embedded in Euclidean space. This set of references focuses on work related to weakly dependent structures, and thus excludes another important branch of the literature on network asymptotics, namely, the literature on exchangeable arrays, because its strong dependence structure differs substantially from the focus of the present paper, both in terms of network configuration and asymptotic theory. For the convenience of readers, however, I list some key theoretical contributions in this area: \citet{GRAHAM202023}, \citet{davezies2021inference}, \citet{menzel2021bootstrap}, \citet{chiang2023inference}, and \citet{graham2024kernel}, among others. A related notion of dependence has also been used for time series \citep[e.g.,][]{babii2022machine}.} In the most recent one of these, \citet*[KMS,][]{KMS2021} establish limit theorems and develop a robust variance estimator for a general class of dependent processes that encompass dependency-graph models in particular. Their framework, grounded in a conditional $\psi$-dependence concept adopted from \citet{doukhan1999new}, offers powerful tools for handling network data and has spurred further research in related fields. Furthermore, the theory and methods introduced by KMS have been widely applied in econometric studies of network models \citep[e.g.,][among others]{leung2022_dependence,gao_ding2023_network,hoshino_yanagi2024_noncompliance}.

Many applications, particularly those involving nonlinear models like limited dependent variable models, demand uniform convergence results. In the context of general classes of M estimators (including maximum likelihood estimators) and generalized method of moments (GMM) estimators, a uniform law of large numbers (ULLN) is crucial for ensuring that the empirical criterion function converges uniformly to its population counterpart. This uniform convergence is fundamental for establishing the consistency, and subsequently the asymptotic normality, of these estimators, as detailed in standard references such as the handbook chapter by \citet[][Section 2]{NeweyMcFadden1994}.

Although KMS offer elegant pointwise limit theorems under network dependence, their results do not directly yield the uniform law of large numbers (ULLN) required for nonlinear estimation. Achieving uniform convergence necessitates controlling not only the individual moments of network-dependent observations but also the fluctuations of the entire process uniformly across the parameter space. 

The main contribution of this paper is to bridge this gap by establishing a novel ULLN under network dependence. The results build on the framework of \citet*{KMS2021}, which leverages model restrictions based on conditional $\psi$-dependence and decay rates of network dependence, concepts that will be briefly reviewed in Section \ref{sec:baseline}. To extend pointwise convergence to uniform convergence, I impose additional regularity conditions including the uniform equicontinuity. The resulting ULLN is then applied to establish the consistency and asymptotic normality of the GMM and M estimators.

This paper was motivated by a practical question raised by a graduate student: ``Can the results of KMS be extended to nonlinear GMM estimation?'' In exploring this question, I identified a critical gap, namely, the lack of a ULLN in the KMS framework, as mentioned above. The purpose of this paper is to address that gap and help bridge the elegant theoretical work of KMS with practical applications in empirical research. It is important to emphasize that the developments presented here rely heavily on the foundational contributions of KMS. While this paper provides a step toward applying their theory to GMM and M estimation, I encourage readers using the results in the present paper to give primary credit to KMS for laying the essential groundwork.

The remainder of the paper is organized as follows. In Section \ref{sec:setup}, I introduce the setup. Section \ref{sec:main} presents the ULLN. Sections \ref{sec:m} and \ref{sec:gmm} introduce M and GMM estimators, respectively, and their asymptotic properties. These two sections also provide practical guidelines. Section \ref{sec:summary} concludes. Mathematical proofs of all the theoretical results are provided in the appendix.

\section{The Setup}\label{sec:setup}
This section introduces the econometric framework.

First, I introduce some basic notations. 
Let \( v, a \in \mathbb{N} \). For any function \( f: \mathbb{R}^{v \times a} \to \mathbb{R} \), define
\[
\|f\|_\infty = \sup_{x \in \mathbb{R}^{v \times a}} |f(x)|
\quad\text{and}\quad
\operatorname{Lip}(f) = \sup_{x \neq y} \frac{|f(x)-f(y)|}{d(x,y)},
\]
where \( d(x,y) \) is a metric on \( \mathbb{R}^{v \times a} \). With these definitions, we introduce the class of uniformly bounded Lipschitz functions:
\[
L_{v,a} = \Bigl\{ f : \mathbb{R}^{v \times a} \to \mathbb{R} \, : \, \|f\|_\infty < \infty \text{ and } \operatorname{Lip}(f) < \infty \Bigr\}.
\]

\subsection{Conditionally $\psi$-Dependent Processes}\label{sec:baseline}
This subsection provides a concise overview of the baseline model introduced in \citet*[KMS,][]{KMS2021} and the notational conventions used in the KMS framework; for a more detailed exposition, please refer to the original paper by KMS.

For each \(n \in \mathbb{N}\), let \(N_n = \{1,2,\ldots,n\}\) denote the set of indices corresponding to the nodes in the network \(G_n\) with the adjacency matrix \(A_n\) whose elements are $0$ and $1$. A link between nodes \(i\) and \(j\) exists if and only if the \((i,j)\)-th entry of \(A_n\) equals one. For each \(n \in \mathbb{N}\), let \(\mathcal{C}_n\) be the \(\sigma\)-algebra with respect to which the adjacency matrix \(A_n\) is measurable. Let \(d_n(i,j)\) denote the network distance between nodes \(i\) and \(j\) in \(N_n\), defined as the length of the shortest path connecting \(i\) and \(j\) in \(G_n\).

For \(a,b \in \mathbb{N}\) and a positive real number \(s\), define
\[
P_n(a,b;s) = \Bigl\{ (A,B) \, : \, A,B \subset N_n,\; |A| = a,\; |B| = b,\; \text{and}\; d_n(A,B) \ge s \Bigr\},
\]
where
\[
d_n(A,B) = \min\{ d_n(i,j) : i \in A,\; j \in B \}.
\]
Thus, each element of \(P_n(a,b;s)\) is a pair of node sets of sizes \(a\) and \(b\) with a distance of at least \(s\) between them.

Consider a triangular array \(\{Y_{n,i}\}_{i \in N_n}\) of random vectors in \(\mathbb{R}^v\). The following definition introduces the notion of conditional \(\psi\)-dependence as provided in KMS.

\begin{definition}[Conditional $\psi$-Dependence; KMS, Definition 2.2]\label{def:psi-dependence}${}$\\
A triangular array $\{Y_{n,i}\}_{i\in N_n}$ is \emph{conditionally $\psi$-dependent given} $\{\mathcal{C}_n\}$ if for each $n\in\mathbb{N}$, there exists a $\mathcal{C}_n$-measurable sequence 
\[
\vartheta_n = \{\vartheta_{n,s}\}_{s\ge 0} \quad \text{with } \vartheta_{n,0}=1,
\]
and a collection of nonrandom functions
\[
\psi_{a,b}: L_{v,a} \times L_{v,b} \to [0,\infty), \quad a,b\in\mathbb{N},
\]
such that for all positive integers $a,b$, for every pair $(A,B)\in P_n(a,b;s)$ with $s>0$, and for all functions $f\in L_{v,a}$ and $g\in L_{v,b}$, the following inequality holds almost surely:
\[
\Bigl|\operatorname{Cov}\bigl(f(Y_{n,A}),\,g(Y_{n,B}) \mid \mathcal{C}_n\bigr)\Bigr| \le \psi_{a,b}(f,g)\,\vartheta_{n,s}.
\]
\end{definition}

As emphasized in KMS, it is important to note that the decay coefficients are generally random, allowing one to accommodate the ``common shocks'' \(\mathcal{C}_n\) present in the network. I now present the following two key assumptions from KMS, which will be employed throughout the present paper.

\begin{assumption}[KMS, Assumption 2.1 (a)]\label{a:kms21}
The triangular array $\{Y_{n,i}\}$ is conditionally $\psi$--dependent given $\{\mathcal{C}_n\}$ with dependence coefficients $\{\vartheta_{n,s}\}$, and there exists a constant $C>0$ such that for all $a,b \in \mathbb{N}$, $f\in L_{v,a}$, and $g\in L_{v,b}$,
    \[
    \psi_{a,b}(f,g) \le C\, ab\, \Bigl(\|f\|_\infty + \operatorname{Lip}(f)\Bigr) \Bigl(\|g\|_\infty + \operatorname{Lip}(g)\Bigr).
    \]
\end{assumption}

For each node $i\in N_n$ in the network for each row $n$ and $s\ge 1$, define
    \[
    N_{n}(i;s) = \{ j\in N_n : d_n(i,j) \leq s\} \quad\text{and}\quad
    N_{n}^\partial(i;s) = \{ j\in N_n : d_n(i,j)=s\},
    \]
representing the number of nodes within and at a distance $s$, respectively.
Then, define the average shell size
    \[
    \delta_{n}^\partial(s) = \frac{1}{n}\sum_{i\in N_n} |N_{n}^\partial(i;s)|.
    \]
With this notation, the following assumption restricts the denseness of the network and the decay rate of dependence with the network distance.

\begin{assumption}[KMS, Assumption 3.2]\label{a:kms32}
The combined effect of network denseness and the decay of dependence is controlled so that
\[
\frac{1}{n}\sum_{s\geq 1} \delta_{ n}^\partial(s)\,\vartheta_{n,s} \to 0 \quad \text{a.s.}
\]
\end{assumption}

I refer readers to the original paper by KMS for detailed discussions of these assumptions, as they are excerpted from KMS. Under these assumptions, along with additional moment and regularity conditions, KMS establish the \textit{pointwise} law of large numbers -- see Proposition~3.1 in their paper.

\subsection{Function Classes}\label{sec:family}
This subsection introduces a parameter-indexed class of functions and imposes additional restrictions to establish the uniform law of large numbers.

Let $\Theta\subset\mathbb{R}^d$ denote a parameter space.
For each $\theta\in\Theta$, let 
    \[
    f(\cdot,\theta): \mathbb{R}^v \to \mathbb{R}
    \]
be a measurable function.
I impose the following conditions on the parameter space $\Theta$ and the function class $\{f(\cdot,\theta):\theta\in\Theta\}$.

\begin{assumption}[Compactness]\label{a:compact}
The parameter space $\Theta\subset\mathbb{R}^d$ is compact.
\end{assumption}

For $p>0$, let $\|f(Y_{n,i},\theta)\|_{\mathcal{C}_n,p}$ denote the conditional $L^p$ norm defined by 
\[
    \|f(Y_{n,i},\theta)\|_{\mathcal{C}_n,p} = \bigl(E\bigl(|f(Y_{n,i},\theta)|^p \mid \mathcal{C}_n\bigr)\bigr)^{1/p}.
\]
With this notation, the following assumption imposes conditions on the function class.

\begin{assumption}[Function Class]\label{a:function}
For each fixed $\theta\in\Theta$:
(i) there exists $\varepsilon > 0$ such that $\sup_{n \in \mathbb{N}}\max_{i\in N_n}\|f(Y_{n,i},\theta)\|_{\mathcal{C}_n,1+\varepsilon} < \infty$ a.s.; and
(ii) $f(\cdot,\theta) \in L_{v,1}$.
\end{assumption}

\noindent
Assumption \ref{a:function} (i) is the bounded moment condition required by Assumption 3.1 of KMS with $f(Y_{n,i},\theta)$ treated as an observation in place of $Y_{n,i}$.
Besides, Assumption \ref{a:function} (ii) imposes the uniform bound and Lipschitz conditions on each function $f(\cdot,\theta)$ in the class.
Taken together, these two components impose restrictions on the behavior of \( f(Y_{n,i}; \theta) \) for `each' \( \theta \), without placing any constraint on the effects of \( \theta \) on it.

For `each' \(\theta \in \Theta\), the pointwise law of large numbers, as stated in Proposition~3.1 of KMS, holds under Assumptions \ref{a:kms21}, \ref{a:kms32}, and \ref{a:function}. I will leverage this pointwise result by KMS as an auxiliary step in establishing the uniform law of large numbers, which requires the following uniform equicontinuity condition in addition.

\begin{assumption}[Uniform Equicontinuity]\label{a:equi}
The function class $\{f(\cdot,\theta): \theta\in\Theta\}$ is uniformly equicontinuous in $\theta$. In particular, there exists a constant $\overline L>0$ such that for all $y$ in the support of $Y_{n,i}$ for all $n$ and for all $\theta,\theta'\in\Theta$, 
\[
|f(y,\theta)-f(y,\theta')| \le \overline L \|\theta-\theta'\|.
\]
\end{assumption}

Assumption \ref{a:equi}, together with Assumption \ref{a:compact}, will allow me to have a finite-net approximation of $f(Y_{n,i},\theta)$ for all $\theta \in \Theta$, as a way to establish the uniform result.

\section{Uniform Law of Large Numbers}\label{sec:main}
I now state the uniform law of large numbers for network-dependent data.

\begin{theorem}[Uniform Law of Large Numbers]\label{theorem:main}
If Assumptions \ref{a:kms21}--\ref{a:equi} are satisfied, then
\[
E\left[\left. \sup_{\theta\in\Theta} \left|\frac{1}{n}\sum_{i\in N_n}\Bigl[f(Y_{n,i},\theta) - E\bigl(f(Y_{n,i},\theta)\mid \mathcal{C}_n\bigr)\Bigr]\right| \ \right\vert \mathcal{C}_n \right] \to 0 \quad\text{a.s.}
\]
\end{theorem}


The next two sections demonstrate how this result can be applied to establish the consistency and asymptotic normality of GMM and M estimators. From this point onward, I focus on the case of a trivial sigma-field $\mathcal{C}_n$ and omit conditioning on it, following the convention in the existing literature \citep[e.g.,][among others]{leung2022_dependence,gao_ding2023_network,hoshino_yanagi2024_noncompliance}, which actually applies the large-sample theory developed by KMS.

To proceed, I introduce few additional notations.
Following KMS (Section 3.1), define
$$
c_n(s,m;k) = \inf_{\alpha>1}[\Delta_n(s,m;k\alpha)]^{1/\alpha}[\delta_n^\partial(s;\alpha/(\alpha-1))]^{1-1/\alpha}
$$
to control the network dependence at distance $s$,
where
\begin{align*}
\Delta_n(s,m;k) &= \frac{1}{n} \sum_{i \in N_n} \max_{j \in N_n^\partial(i;s)} |N_n(i;m) \backslash N_n(j;s-1)|^k
\quad\text{and}\\
\delta_n^\partial(s;k) &= \frac{1}{n} \sum_{i \in N_n} |N_n^\partial(i;s)|^k.
\end{align*}
Recall that $N_n(i;s)$ and $N_n^\partial(i;s)$ are defined in Section \ref{sec:baseline}.
I refer readers to KMS (Section 3.1) for discussions of these objects and the roles which they play.
Finally, let $\lambda_{\min}(A)$ denote the minimum eigenvalue of square matrix $A$.

\section{M Estimation}\label{sec:m}

Let $Q(\cdot)$ and $Q_n(\cdot)$ be the population and sample criterion functions for M estimation, defined on $\Theta$ by
$$
Q(\theta) = E\bigl(f(Y_{n,i},\theta)\bigr)
\quad\text{and}\quad
Q_n(\theta) = \frac{1}{n}\sum_{i \in N_n} f(Y_{n,i},\theta),
$$
respectively.\footnote{\label{foot:q_n}I consider the case where the population criterion is independent of $n$, but a slight modification of the assumptions can accommodate settings where the population criterion depends on $n$.}
The M estimator is defined as
\[
\hat{\theta}_{M} \in \arg\max_{\theta\in\Theta} Q_n(\theta).
\]
In the pesudo maximum likelihood estimation (PMLE) framework, $f(Y_{n,i},\theta)$ corresponds to the logarithm of the marginal density function of $Y_{i,n}$ given the parameter $\theta$.

\subsection{Consistency of the M Estimator}\label{sec:m:consistency}

Suppose that the population criterion satisfies the following condition.

\begin{assumption}[Identification for M Estimation]\label{ass:m:identification}
The objective function $Q(\cdot)$ is continuous on $\Theta$ and there exists a unique \(\theta_0 \in \text{int}(\Theta)\) such that
$
\{\theta_0\} = \arg\max_{\theta \in \Theta} Q(\theta).
$
\end{assumption}

With this identification condition, the standard argument based on \citet[][Theorem~2.1]{NeweyMcFadden1994}, for example, yields the consistency $\hat\theta_{M} \stackrel{p}{\rightarrow} \theta_0$ by the uniform law of large numbers (my Theorem \ref{theorem:main}).
Let me state this conclusion formally as the following corollary to Theorem \ref{theorem:main}.

\begin{corollary}[Consistency of the M Estimator]\label{cor:m}
If Assumptions \ref{a:kms21}--\ref{ass:m:identification} hold, then $\hat\theta_{M} \stackrel{p}{\rightarrow} \theta_0$.
\end{corollary}


\subsection{Asymptotic Normality of the M Estimator}\label{sec:m:normality}

To establish the asymptotic normality, the following two assumptions are used in addition.

\begin{assumption}\label{a:m:normal}
Let $\Sigma_n = \text{Var}\left(\sum_{i \in N_n} \nabla_{\theta} f(Y_{n,i},\theta_0)\right)$ for each $n \in \mathbb{N}$.
There exists some $p > 4$ such that:
(i) $\sup_{n \in \mathbb{N}}\max_{i \in N_n} E[\|\nabla_{\theta}f(Y_{n,i},\theta_0)\|^p]<\infty$;
(ii) $\sup_{n \geq 1} \max_{s \geq 1} \vartheta_{n,s} < \infty$ for each $k=1,2$;
(iii) $\frac{n}{\lambda_{\min}(\Sigma_n)^{2+k}} \sum_{s \geq 0} c_n(s,m_n;k)\vartheta_{n,s}^{1-(2+k)/p} \rightarrow 0$ and
$\frac{n^2 \vartheta_{n,m_n}^{1-1/p}}{\lambda_{\min}(\Sigma_n)} \rightarrow 0$; and
(iv) $n^{-2}\Sigma_n \rightarrow \Sigma$.
\end{assumption}

This assumption is invoked to directly obtain the CLT of KMS (their Theorem 3.2) for establishing the asymptotic normality of $\sqrt{n} a^\top \nabla_\theta Q_n(\theta_0)$ for a vector $a$ such that $\|a\|=1$.\footnote{While KMS consider multivariate random variables, their CLT result is stated for univariate cases.}
With our focus on the trivial sigma-field $\mathcal{C}_n$,
part (i) of Assumption \ref{a:m:normal} implies Assumption 3.3 of KMS,
part (ii) implies Assumption 2.1 (b) of KMS, and
part (iii) implies Assumption 3.4 of KMS by Rayleigh quotient.
Part (iv) requires that the vairance of the mean score in the griangular array converges.
We refer readers to KMS for discussions of these conditions.

\begin{assumption}\label{a:m:hessian}
(i) $\theta \mapsto f(y,\theta)$ is twice differentiable for all $y$ on an open set containing the support of $Y_{n,i}$. Let Assumptions \ref{a:function}--\ref{a:equi} be satisfied with each element of $\nabla_{\theta\theta}f(\cdot,\theta)$ in place of $f(\cdot,\theta)$.
(ii) There exists a function $h: \mathbb{R}^v \rightarrow \mathbb{R}_+$ such that $\sup_{n \in \mathbb{N}}\max_{i \in N_n} E[h(Y_{n,i})] < \infty$ and $h$ dominates $y \mapsto \nabla_{\theta\theta} f(y,\theta)$ for all $\theta \in \Theta$.
(iii) $H := \nabla_{\theta\theta} Q(\theta_0)$ is non-singular
\end{assumption}

Parts (i)--(ii) of Assumption \ref{a:m:hessian}, together with Assumptions \ref{a:kms21}, \ref{a:kms32}, and  \ref{a:compact}, are used to invoke the uniform law of large numbers (Theorem \ref{theorem:main}) on the Hessian: $\sup_{\theta \in \Theta} |\nabla_{\theta\theta} Q_n(\theta) - \nabla_{\theta\theta} Q(\theta)| \rightarrow 0$ a.s., where the equicontinuity in part (i) and the $L^1$ dominance in part (ii) allow the dominated convergence theorem to yield $\nabla_{\theta\theta} Q(\theta) = E\bigl(\nabla_{\theta\theta} f(Y_{n,i},\theta)\bigr)$, which is guaranteed to be a continuous function of $\theta$.
Further, part (iii) ensures that its limit is invertible at $\theta_0$.

Now, combining the CLT of KMS (their Theorem 3.2) with my Theorem \ref{theorem:main} and Corollary \ref{cor:m}, we obtain the following asymptotic normality result through checking the conditions of \citet[][Theorem 3.1]{NeweyMcFadden1994}.

\begin{corollary}[Asymptotic Normality of the M Estimator]\label{cor:m:normal}
If Assumptions \ref{a:kms21}--\ref{a:m:hessian} hold, then $\sqrt{n}\left(\hat\theta_M - \theta_0\right) \stackrel{d}{\rightarrow} N\left(0,H^{-1} \Sigma H^{-1}\right)$.
\end{corollary}


\subsection{Guide in Practice for M Estimation}\label{sec:m:practical}

The current section presents the practical procedure to implement an M estimation under network dependence.

First, obtain the estimate
\[
\hat{\theta}_{M} \in \arg\max_{\theta\in\Theta} \frac{1}{n}\sum_{i \in N_n} f(Y_{n,i},\theta).
\]

Second, adapting the network HAC estimation procedure of KMS (Section 4) to the present framework of M estimation, compute the network-robust variance estimate
\[
\hat \Sigma = \sum_{s \geq 0} \omega(s/b_n) \cdot \frac{1}{n} \sum_{i \in N_n} \sum_{j \in N_n^\partial(i;s)} \left( \nabla_\theta f(Y_{n,i},\hat\theta) \right) \left( \nabla_\theta f(Y_{n,i},\hat\theta) \right)^\top
\]
for the score,
where $\omega(\cdot)$ denotes a kernel function\footnote{The kernel $\omega: \mathbb{R} \rightarrow [-1,1]$ satisfies $\omega(0) = 1$, $\omega(z) = 0$ for $|z| > 1$, and $\omega(z) = \omega(-z)$ for all $z \in \mathbb{R}$.} and $b_n$ is a bandwidth parameter.

For example, using the Parzen kernel,
\[
\omega(u) = 
\begin{cases}
1 - 6u^2 + 6|u|^3, & \text{if } 0 \le |u| \le \frac{1}{2} \\
2(1 - |u|)^3, & \text{if } \frac{1}{2} < |u| \le 1 \\
0, & \text{if } |u| > 1
\end{cases}
\]
KMS demonstrate that the following bandwidth choice performs well in simulations:\footnote{That said, the optimal choice of bandwidth should remain an important direction for future research.}
\[
b_n = \frac{2\log n}{\log\left( \max\{\hat\delta_n^\partial(1), 1.05\} \right)},
\]
where $\hat\delta_n^\partial(1)$ denotes the average degree of the observed network.

Finally, compute the Hessian estimator
\[
\hat H = \frac{1}{n} \sum_{i \in N_n} \nabla_{\theta\theta} f(Y_{n,i},\hat\theta).
\]
Note that even in the PMLE framework, the information equality (which is established under i.i.d. sampling) may not hold in general under network dependence.

\section{GMM Estimation}\label{sec:gmm}

Let $f(\cdot,\cdot)$ denote the moment function such that the true parameter vector $\theta_0 \in \Theta$ satisfies the moment equality
$$
E\bigl(f(Y_{n,i},\theta_0)\bigr) = 0.
$$
Define the sample moment function by
\[
\bar{f}_n(\theta) = \frac{1}{n}\sum_{i=1}^n f(Y_{n,j}, \theta).
\]
For any sequence \(W_n\) of positive definite weighting matrices (which may depend on the data) converging in probability to a positive definite matrix $W$, the GMM estimator is defined as
\begin{align*}
&\hat{\theta}_{GMM} = \arg\min_{\theta\in\Theta} Q_n(\theta),
\\
&\text{where} \quad Q_n(\theta) = \bar{f}_n(\theta)^\top W_n\, \bar{f}_n(\theta).
\end{align*}
We can define the population criterion by\footnote{A similar remark to Footnote~\ref{foot:q_n} applies here.}
\[
Q(\theta) = E\bigl(f(Y_{n,i},\theta)\bigr)^\top W E\bigl(f(Y_{n,i},\theta)\bigr).
\]

\subsection{Consistency of the GMM Estimator}\label{sec:gmm:consistency}

Suppose that the population moment satisfies the following condition.

\begin{assumption}[Identification for GMM Estimation]\label{ass:gmm:identification}
The objective function $Q(\cdot)$ is continuous on $\Theta$ and there exists a unique \(\theta_0\in\text{int}(\Theta)\) such that
\[
E\bigl(f(Y_{n,i},\theta)\bigr)=0 \quad \text{if and only if} \quad \theta=\theta_0.
\]
\end{assumption}

With this identification condition, the standard argument based on \citet[][Theorem~2.1]{NeweyMcFadden1994}, for example, yields the consistency $\hat\theta_{GMM} \stackrel{p}{\rightarrow} \theta_0$ by the uniform law of large numbers (my Theorem \ref{theorem:main}).
Let me state this conclusion formally as the following corollary to Theorem \ref{theorem:main}.

\begin{corollary}[Consistency of the GMM Estimator]\label{cor:gmm}
If Assumptions \ref{a:kms21}--\ref{a:equi}, and \ref{ass:gmm:identification} hold, then $\hat\theta_{GMM} \stackrel{p}{\rightarrow} \theta_0$.
\end{corollary}


\subsection{Asymptotic Normality of the GMM Estimator}\label{sec:gmm:normality}

To establish the asymptotic normality, the following two assumptions are used in addition.

\begin{assumption}\label{a:gmm:normal}
Let $\Omega_n = \text{Var}\left(\sum_{i \in N_n} f(Y_{n,i},\theta_0)\right)$ for each $n \in \mathbb{N}$.
There exists some $p > 4$ such that:
(i) $\sup_{n \in \mathbb{N}}\max_{i \in N_n} E[\|f(Y_{n,i},\theta_0)\|^p]<\infty$;
(ii) $\sup_{n \geq 1} \max_{s \geq 1} \vartheta_{n,s} < \infty$ for each $k=1,2$;
(iii) $\frac{n}{\lambda_{\min}(\Omega_n)^{2+k}} \sum_{s \geq 0} c_n(s,m_n;k)\vartheta_{n,s}^{1-(2+k)/p} \rightarrow 0$ and
$\frac{n^2 \vartheta_{n,m_n}^{1-1/p}}{\lambda_{\min}(\Omega_n)} \rightarrow 0$; and
(iv) $n^{-2}\Omega_n \rightarrow \Omega$.
\end{assumption}

\begin{assumption}\label{a:gmm:hessian}
(i) $\theta \mapsto f(y,\theta)$ is differentiable for all $y$ on an open set containing the support of $Y_{n,i}$. Let Assumptions \ref{a:function}--\ref{a:equi} be satisfied with each element of $D_{\theta}f(\cdot,\theta)$ in place of $f(\cdot,\theta)$.
(ii) There exists a function $g: \mathbb{R}^v \rightarrow \mathbb{R}_+$ such that $\sup_{n \in \mathbb{N}}\max_{i \in N_n} E[g(Y_{n,i})] < \infty$ and $g$ dominates $y \mapsto D_{\theta} f(y,\theta)$ for all $\theta \in \Theta$.
(iii) $G := D_{\theta} E\left(f(Y_{n,i},\theta_0)\right)$ is non-singular
\end{assumption}

Assumptions \ref{a:gmm:normal} and \ref{a:gmm:hessian} are analogous to Assumptions \ref{a:m:normal} and \ref{a:m:hessian}, respectively, and hence similar discussions apply, which are omitted here to avoid repetitions.

Combining the CLT of KMS (their Theorem 3.2) with my Theorem \ref{theorem:main} and Corollary \ref{cor:gmm}, we obtain the following asymptotic normality result through checking the conditions of \citet[][Theorem 3.2]{NeweyMcFadden1994}.

\begin{corollary}[Asymptotic Normality of the GMM Estimator]\label{cor:gmm:normal}
If Assumptions \ref{a:kms21}--\ref{a:equi} and \ref{ass:gmm:identification}--\ref{a:gmm:hessian} hold, then 
$
\sqrt{n}\left(\hat\theta_{GMM} - \theta_0\right)
\stackrel{d}{\rightarrow}
N\left(0,(G^\top W G)^{-1} G^\top W \Omega W G (G^\top W G)^{-1}\right).
$
\end{corollary}


\subsection{Guide in Practice for GMM Estimation}\label{sec:gmm:practical}

The current section presents the practical procedure to implement an GMM estimation under network dependence.

First, obtain the estimate
\begin{align*}
&\hat{\theta}_{GMM} = \arg\min_{\theta\in\Theta} \left(\frac{1}{n}\sum_{i=1}^n f(Y_{n,j}, \theta)\right)^\top W_n\, \left(\frac{1}{n}\sum_{i=1}^n f(Y_{n,j}, \theta)\right).
\end{align*}
Second, adapting the network HAC estimation procedure of KMS (Section 4) to the present framework of GMM estimation, compute the network-robust variance estimate
\[
\hat \Omega = \sum_{s \geq 0} \omega(s/b_n) \cdot \frac{1}{n} \sum_{i \in N_n} \sum_{j \in N_n^\partial(i;s)} f(Y_{n,i},\hat\theta) f(Y_{n,i},\hat\theta)^\top,
\]
where $\omega(\cdot)$ is a kernel function and $b_n$ is a bandwidth parameter.
See Section \ref{sec:m:practical} for further discussions of $\omega(\cdot)$ and $b_n$.

Finally, compute the gradient estimator
\[
\hat G = \frac{1}{n} \sum_{i \in N_n} D_\theta f(Y_{n,i},\hat\theta).
\]
As usual, one may iterate the above procedure to implement the two-step GMM estimation.

\section{Summary and Discussions}\label{sec:summary}

This paper establishes the asymptotic properties of GMM and M estimators under network dependence. As a key step toward this goal, I extend the law of large numbers from \citet*[Proposition 3.1]{KMS2021} to a novel uniform law of large numbers (ULLN), stated in Theorem \ref{theorem:main}. Since the consistency of nonlinear estimators, such as GMM and M estimators, requires uniform convergence of the criterion functions, this result lays the foundation for proving their consistency and, subsequently, their asymptotic normality. For completeness, Sections \ref{sec:m} and \ref{sec:gmm} present full sets of assumptions under which these asymptotic properties hold for the M and GMM estimators, respectively.

As already mentioned in the introduction, this paper originated from a practical question posed by a graduate student: ``Can the results of KMS be applied to nonlinear GMM estimation?'' In addressing this question, I identified a key gap, namely, the absence of a ULLN in KMS as discussed above. This paper was written to bridge that gap and connect the elegant theory of KMS with the needs of empirical practitioners. That said, the results presented here build heavily on KMS, and much of the foundational work and theoretical development should be credited to their contribution. Accordingly, even if readers use the results presented in this paper in the context of GMM and M estimation, I strongly encourage them to give primary credit to KMS, whose work has done most of the heavy lifting.

\bibliographystyle{chicago} 
\bibliography{mybib}

\appendix
\section*{Appendix}

The appendix  collects proofs of the theoretical results presented in the main text.
Specifically, Appendix \ref{sec:theorem:main}, \ref{sec:cor:m}, \ref{sec:cor:m:normal}, \ref{sec:cor:gmm}, and \ref{sec:cor:gmm:normal} present proofs of Theorem \ref{theorem:main}, Corollary \ref{cor:m}, Corollary \ref{cor:m:normal}, Corollary \ref{cor:gmm}, and Corollary \ref{cor:gmm:normal} respectively.

\section{Proofs}\label{sec:proofs}

\subsection{Proof of Theorem \ref{theorem:main}}\label{sec:theorem:main}
\begin{proof}
By Assumption \ref{a:compact}, for any $\delta>0$ there exists a finite $\delta$--net $\{\theta_1,\theta_2,\ldots,\theta_J\}\subset \Theta$ such that for every $\theta\in\Theta$, there exists some $\theta_j \in \{\theta_1,\theta_2,\ldots,\theta_J\}$ with
\[
\|\theta-\theta_j\|\le\delta.
\]
For an arbitrary $\theta\in\Theta$, let $\theta_j(\theta) \in \{\theta_1,\theta_2,\ldots,\theta_J\}$ be a net point with $\|\theta-\theta_j(\theta)\|\le \delta$. Decompose
\[
\frac{1}{n}\sum_{i\in N_n}\Bigl[f(Y_{n,i},\theta)-E\bigl(f(Y_{n,i},\theta)\mid \mathcal{C}_n\bigr)\Bigr]
= A_n(\theta) + B_n(\theta) + C_n(\theta),
\]
where the three components on the right-hand side are:
\begin{align*}
A_n(\theta)=&\frac{1}{n}\sum_{i\in N_n}\Bigl[f(Y_{n,i},\theta)-f(Y_{n,i},\theta_j(\theta))\Bigr],
\\
B_n(\theta)=&\frac{1}{n}\sum_{i\in N_n}\Bigl[f(Y_{n,i},\theta_j(\theta))-E\bigl(f(Y_{n,i},\theta_j(\theta))\mid \mathcal{C}_n\bigr)\Bigr],
\qquad\text{and}
\\
C_n(\theta)=&\frac{1}{n}\sum_{i\in N_n}E\Bigl[f(Y_{n,i},\theta_j(\theta))-f(Y_{n,i},\theta)\mid \mathcal{C}_n\Bigr].
\end{align*}

By Assumption \ref{a:equi}, for every $n,i$, we have
\[
|f(Y_{n,i},\theta)-f(Y_{n,i},\theta_j(\theta))| \le \overline L\,\|\theta-\theta_j(\theta)\| \le \overline L\,\delta.
\]
Taking sample mean gives
\[
|A_n(\theta)| \le
\frac{1}{n}\sum_{i \in N_n} |f(Y_{n,i},\theta)-f(Y_{n,i},\theta_j(\theta))|
\le \overline L\,\delta.
\]
Similarly,
\[
|C_n(\theta)| \le
\frac{1}{n}\sum_{i \in N_n} E\Bigl[ |f(Y_{n,i},\theta_j(\theta))-f(Y_{n,i},\theta)| \, \mid \mathcal{C}_n\Bigr]
\le \overline L\,\delta.
\]
Thus, it follows that
\[
|A_n(\theta) + C_n(\theta)| \le 2\overline L\,\delta.
\]

Applying Proposition~3.1 of \citet*{KMS2021} under my Assumptions \ref{a:kms21}, \ref{a:kms32}, and \ref{a:function} for each fixed $\theta_j \in \{\theta_1,\theta_2,\ldots,\theta_J\}$ gives
\[
\left\| \frac{1}{n}\sum_{i\in N_n}\Bigl[f(Y_{n,i},\theta_j)-E\bigl(f(Y_{n,i},\theta_j)\mid \mathcal{C}_n\bigr)\Bigr] \right\|_{\mathcal{C}_n,2} \to 0 \quad \text{a.s.}
\]
Since the $\delta$--net $\{\theta_1,\ldots,\theta_J\}$ is finite, it follows that
\begin{equation}\label{eq:mn_zero}
S_n := \max_{1\le j\le J}\left\|\frac{1}{n}\sum_{i\in N_n}\Bigl[f(Y_{n,i},\theta_j)-E\bigl(f(Y_{n,i},\theta_j)\mid \mathcal{C}_n\bigr)\Bigr]\right\|_{\mathcal{C}_n,2} \to 0 \quad \text{a.s.}
\end{equation}

For any $\theta\in\Theta$,
\[
\left|\frac{1}{n}\sum_{i\in N_n}\Bigl[f(Y_{n,i},\theta)-E\bigl(f(Y_{n,i},\theta)\mid \mathcal{C}_n\bigr)\Bigr]\right|
\le |B_n(\theta)| + |A_n(\theta)+C_n(\theta)|,
\]
and thus
\begin{align*}
&\sup_{\theta\in\Theta}\left|\frac{1}{n}\sum_{i\in N_n}\Bigl[f(Y_{n,i},\theta)-E\bigl(f(Y_{n,i},\theta)\mid \mathcal{C}_n\bigr)\Bigr]\right|
\\
\le& \max_{1\le j\le J}\left|\frac{1}{n}\sum_{i\in N_n}\Bigl[f(Y_{n,i},\theta_j)-E\bigl(f(Y_{n,i},\theta_j)\mid \mathcal{C}_n\bigr)\Bigr]\right| + 2\overline L\,\delta.
\end{align*}
Take $E[ \ \cdot \ |\mathcal{C}_n]$ on both sides to get
\begin{align*}
&E\left[\left. \sup_{\theta\in\Theta}\left|\frac{1}{n}\sum_{i\in N_n}\Bigl[f(Y_{n,i},\theta)-E\bigl(f(Y_{n,i},\theta)\mid \mathcal{C}_n\bigr)\Bigr]\right| \ \right\vert \mathcal{C}_n\right]
\\
\le& E\left[\left. \max_{1\le j\le J}\left|\frac{1}{n}\sum_{i\in N_n}\Bigl[f(Y_{n,i},\theta_j)-E\bigl(f(Y_{n,i},\theta_j)\mid \mathcal{C}_n\bigr)\Bigr]\right| \ \right\vert \mathcal{C}_n\right] + 2\overline L\,\delta
\le J S_n + 2\overline L\,\delta,
\end{align*}
where the second inequality uses the Cauchy–Schwarz inequality.

Since for every fixed $\delta>0$ the maximum $S_n$ over the $\delta$-net converges to 0 almost surely by \eqref{eq:mn_zero}, it follows that
\[
E\left[\left. \sup_{\theta\in\Theta}\left|\frac{1}{n}\sum_{i\in N_n}\Bigl[f(Y_{n,i},\theta)-E\bigl(f(Y_{n,i},\theta)\mid \mathcal{C}_n\bigr)\Bigr]\right| \ \right\vert \mathcal{C}_n\right]
\le 2\overline L\,\delta + o(1)
\quad\text{a.s.}
\]
Because $\delta>0$ is arbitrary, this shows that
\[
E\left[\left. \sup_{\theta\in\Theta}\left|\frac{1}{n}\sum_{i\in N_n}\Bigl[f(Y_{n,i},\theta)-E\bigl(f(Y_{n,i},\theta)\mid \mathcal{C}_n\bigr)\Bigr]\right| \ \right\vert \mathcal{C}_n \right] \to 0 \quad \text{a.s.}
\]
as claimed in the statement of the theorem.
\end{proof}

\subsection{Proof of Corollary \ref{cor:m}}\label{sec:cor:m}
\begin{proof}
I am going to check the four conditions of \citet[NM,][Theorem 2.1]{NeweyMcFadden1994}.
Assumption \ref{ass:m:identification} implies condition (i) in NM.
Assumption \ref{a:compact} implies condition (ii) in NM.
Condition (iii) in NM is directly assumed in Assumption \ref{ass:m:identification}.
Theorem \ref{theorem:main} under Assumptions \ref{a:kms21}--\ref{a:equi} implies condition (iv) in NM.
Therefore, the claim of the corollary follows by Theorem 2.1 of NM.
\end{proof}

\subsection{Proof of Corollary \ref{cor:m:normal}}\label{sec:cor:m:normal}

\begin{proof}
I am going to check the conditions of \citet[NM,][Theorem 3.1]{NeweyMcFadden1994}.
Corollary \ref{cor:m} under Assumptions \ref{a:kms21}--\ref{ass:m:identification} implies the consistency
$\hat\theta_{M} \stackrel{p}{\rightarrow} \theta_0$ required by Theorem 3.1 in NM.
Condition (i) in NM is directly stated in Assumption \ref{ass:m:identification}.
Theorem 3.2 of KMS under my Assumptions \ref{a:kms21} and \ref{a:m:normal}, together with the Wold device, yields
$
\sqrt{n} \nabla_{\theta} Q_n(\theta_0) \stackrel{d}{\rightarrow} N(0,\Sigma).
$
This shows that condition (iii) in NM is satisfied.
Note that, by the equicontinuity in Assumption \ref{a:m:hessian} (i) and the $L^1$ dominance in Assumption \ref{a:m:hessian} (ii), the dominated convergence theorem yields $\nabla_{\theta\theta} Q(\theta) = E\bigl(\nabla_{\theta\theta} f(Y_{n,i},\theta)\bigr)$, and this derivative function is guaranteed to be a continuous function of $\theta$.
Therefore, Assumptions \ref{a:kms21}, \ref{a:kms32}, \ref{a:compact}, and \ref{a:m:hessian} (i)--(ii) imply
$
\sup_{\theta \in \Theta} \left\vert \nabla_{\theta\theta} Q_n(\theta) - \nabla_{\theta\theta} Q(\theta) \right\vert \rightarrow 0
$ 
a.s.
by Theorem \ref{theorem:main}.
This shows that condition (iv) in NM is satisfied.
Finally, condition (v) of NM is directly stated in Assumption \ref{a:m:hessian} (iii).
Thus, the proof of the corollary is complete.
\end{proof}

\subsection{Proof of Corollary \ref{cor:gmm}}\label{sec:cor:gmm}
\begin{proof}
I am going to check the four conditions of \citet[NM,][Theorem 2.1]{NeweyMcFadden1994}.
Assumption \ref{ass:gmm:identification} and the positive definiteness of $W$ imply condition (i) in NM.
Assumption \ref{a:compact} implies condition (ii) in NM.
Condition (iii) in NM is directly assumed in Assumption \ref{ass:gmm:identification}.
Since $W_n \stackrel{p}{\rightarrow} W$ where $W$ is positive definite, Theorem \ref{theorem:main} under Assumptions \ref{a:kms21}--\ref{a:equi} implies condition (iv) in NM.
Therefore, the claim of the corollary follows by Theorem 2.1 of NM.
\end{proof}

\subsection{Proof of Corollary \ref{cor:gmm:normal}}\label{sec:cor:gmm:normal}

\begin{proof}
I am going to check the conditions of \citet[NM,][Theorem 3.2]{NeweyMcFadden1994}.
Corollary \ref{cor:gmm} under Assumptions \ref{a:kms21}--\ref{a:equi}, and \ref{ass:gmm:identification} implies the consistency
$\hat\theta_{M} \stackrel{p}{\rightarrow} \theta_0$ required by Theorem 3.2 in NM.
Condition (i) in NM is directly stated in Assumption \ref{ass:gmm:identification}.
Theorem 3.2 of KMS under my Assumptions \ref{a:kms21} and \ref{a:gmm:normal}, together with the Wold device, yields
$
\sqrt{n} \bar f_n(\theta_0) \stackrel{d}{\rightarrow} N(0,\Omega).
$
This shows that condition (iii) in NM is satisfied.
Note that, under the equicontinuity in Assumption \ref{a:gmm:hessian} (i) and the $L^1$ dominance in Assumption \ref{a:gmm:hessian} (ii), the dominated convergence theorem yields $D_{\theta}E\bigl(f(Y_{n,i},\theta)\bigr) = E\bigl(D_{\theta} f(Y_{n,i},\theta)\bigr)$, and this derivative function is guaranteed to be a continuous function of $\theta$.
Therefore, Assumptions \ref{a:kms21}, \ref{a:kms32}, \ref{a:compact}, and \ref{a:gmm:hessian} (i)--(ii) imply
$
\sup_{\theta \in \Theta} \left\vert D_{\theta} \bar f_n(\theta) - G(\theta) \right\vert \rightarrow 0
$ 
a.s.
by Theorem \ref{theorem:main}, where $G(\theta) := D_{\theta} E\left(f(Y_{n,i},\theta)\right)$.
This shows that condition (iv) in NM is satisfied.
Finally, condition (v) of NM is directly stated in Assumption \ref{a:gmm:hessian} (iii).
Thus, the proof of the corollary is complete.
\end{proof}

\end{document}